\pdfoutput=1
\documentclass{llncs}

\usepackage{amsfonts,amsmath}

\usepackage{graphicx}
\usepackage{mathptmx}
\usepackage{tikz-qtree}
\usepackage{pgfplots}
\usepackage{filecontents}
\usepackage{csvsimple}
\usepackage{amsfonts,amsmath}
\usepackage{url}
\usepackage{verbatim}
\usepackage{color}

\definecolor{lgray}{gray}{0.95}
\definecolor{lblue}{rgb}{0.90,0.90,1.00}
\definecolor{lyellow}{rgb}{1.00,1.00,0.70}

\usepackage{listings}
\newtheorem{prop}{Proposition}

\newtheorem{ex}{Example}

\newenvironment{codex}{\small\verbatim}{\endverbatim\normalsize}

\lstnewenvironment{code}
    {\lstset{}%
      \csname lst@SetFirstLabel\endcsname}
    {\csname lst@SaveFirstLabel\endcsname}
\newcommand{\BI}[0]{\begin{itemize}}
\newcommand{\EI}[0]{\end{itemize}}
\newcommand{\I}[0]{\item}
\newcommand{\BE}[0]{\begin{enumerate}}
\newcommand{\EE}[0]{\end{enumerate}}

\newcommand{\BX}[0]{\begin{ex}}
\newcommand{\EX}[0]{\end{ex}}
\newcommand{\BV}[0]{\begin{verbatim}}
\newcommand{\EV}[0]{\end{verbatim}}

\newcommand{\BF}[0]{\begin{filecontents*}{data.csv}}
  
\newcommand{\LPIC}[4]{
\begin{figure}[htbp]
\begin{center}
\begin{tikzpicture}
\begin{semilogyaxis}[xlabel={#1},ylabel={#2} (log. scale)]
\addplot table [col sep=comma] {data.csv};
\end{semilogyaxis}
\end{tikzpicture}
\end{center}
\caption{#3}
\label{#4}
\end{figure}
}  

\newcommand{\LPICS}[6]{
\begin{figure}[htbp]
\begin{center}
\begin{tikzpicture}
\begin{semilogyaxis}[xlabel={#1},ylabel={#2} (log. scale),
		legend entries={#3,#4},
		legend style={
      at={(0.50,1.03)},
			anchor=south
		}
]
\addplot table [x=a,y=b,col sep=comma] {data.csv};
\addplot table [x=a,y=c,col sep=comma] {data.csv};
\end{semilogyaxis}
\end{tikzpicture}
\end{center}
\caption{#5}
\label{#6}
\end{figure}
}

\def \bscale1 {0.25}
\def \bscale {0.25}

\begin{document}

\title{ 
On Uniquely Closable and Uniquely Typable Skeletons of Lambda Terms 
}

\author{Olivier Bodini\inst{1} \and Paul Tarau\inst{2}}

\institute{
   {Laboratoire d'Informatique de Paris-Nord}\\
   {UMR CNRS 7030}\\
   {\em olivier.bodini@lipn.univ-paris13.fr}
\and
   {Department of Computer Science and Engineering}\\
   {University of North Texas}\\
   {\em paul.tarau@unt.edu}
   }

\maketitle

\begin{abstract}
Uniquely closable skeletons of lambda terms are Motzkin-trees that
predetermine the unique closed lambda term that can be obtained by labeling their leaves with
 de Bruijn indices. Likewise, uniquely typable skeletons of closed lambda terms
predetermine the unique simply-typed lambda term that can be obtained 
by labeling their leaves with de Bruijn indices.

We derive, through a sequence of logic program transformations, efficient
code for their combinatorial generation and study their statistical properties.

As a result, we obtain context-free grammars describing closable and uniquely closable skeletons
of lambda terms, opening the door for their in-depth study with tools from analytic combinatorics.

Our empirical study of the more difficult case of (uniquely) typable terms reveals some
interesting open problems about their density and asymptotic behavior. 

As a connection between the two classes of terms, we also show that uniquely typable closed lambda term skeletons of size $3n+1$ are  in a bijection with binary trees of size $n$.

{\em {\bf Keywords:}
deriving  efficient logic programs,
logic programming and computational mathematics,
combinatorics of lambda terms,
inferring simple types,
uniquely closable lambda term skeletons,
uniquely typable lambda term skeletons.
}
\end{abstract}

\section{Introduction}

The study of combinatorial properties of lambda terms has theoretical ramifications
ranging from their connection to proofs in intuitionistic logic 
via the Curry-Howard correspondence \cite{howard:formulaeastypes:hbc:80} 
and their role as a foundation of
Turing-complete as well as expressive but terminating computations in the case of simply typed lambda terms \cite{bar93}. At the same time, lambda terms are
used in the internal representations of
compilers for functional programming languages and proof assistants,
for which the generation of large random lambda terms helps
with automated testing \cite{palka11}. 

This paper focuses on
binary-unary trees that are obtained from lambda terms
in de Bruijn form, represented as trees, 
by erasing the de Bruijn indices labeling variables at their leaves.
Such ``skeletons'' of the lambda terms turn out to predetermine some
non-trivial properties the lambda terms they host, e.g.,
if such terms are closed or simply-typed.
Of particular interest are the cases when unique such terms
exist.

Our declarative meta-language is Prolog, which turns out to provide
everything we need: easy combinatorial generation via backtracking
over the set of all answers, specified as a Definite Clause Grammar (DCG)
that enforces size constraints and allows  placing more complex constraints
at points in the code where they ensure the earliest possible pruning
of the search space.  

Our meta-language also facilitates program transformations
that allow us to derive step-by-step faster programs as well as simpler
expressions of the underlying combinatorial mechanisms, e.g., a context-free
grammar in the case of uniquely closable skeletons, that in turn makes them
amenable to study with powerful techniques from analytical combinatorics.

The paper is organized as follows.
Section \ref{ct} describes generators for closed lambda terms and their Motzkin-tree skeletons.
Section \ref{cs} introduces closable skeletons and studies their statistical properties.
Section \ref{uc} derives  algorithms (including a CF-grammar) for efficient generation of uniquely closable skeletons.
Section \ref{utnt} discusses typable and untypable closed skeletons.
Section \ref{ucut} introduces uniquely typable closed skeletons, studies the special case of uniquely closable and uniquely typable skeletons and establishes 
their connection to members of the Catalan family of combinatorial objects.
Section \ref{rels} overviews related work and section
\ref{concl} concludes the paper.

\begin{codeh}
:-include('lpgen/stats.pro').
\end{codeh}

The paper is structured as a literate Prolog program to facilitate an easily replicable, concise and declarative expression of our concepts and algorithms. 

The code extracted from the paper, tested with SWI-Prolog \cite{swi} version {\tt 7.5.3}, is available at: 
\url{http://www.cse.unt.edu/~tarau/research/2017/uct.pro} .

\section{Closed Lambda Terms and their Motzkin-tree Skeletons} \label{ct}

A {\em Motzkin tree} (also called binary-unary tree) is a rooted ordered tree
built from binary nodes, unary nodes and leaf nodes.
Thus the set of  Motzkin trees can be seen as
the free algebra generated by 
the constructors {\tt v/0}, {\tt l/1} and {\tt a/2}.

We define lambda terms in de Bruijn form \cite{dbruijn72} 
as the free algebra generated by 
the constructors {\tt l/1}, and {\tt a/2}
and leaves labeled with natural numbers wrapped
with the constructor {\tt v/1}.

A lambda term in de Bruijn form 
is {\em closed} if for each
of its de Bruijn indices it exists
a lambda binder to which it points,
on the path to the root of the tree
representing the term. They are counted by sequence {\bf A135501} in \cite{intseq}.

The predicate {\tt closedTerm/2} specifies an all-terms generator, which, given a natural number {\tt N}
backtracks one member {\tt X} at a time, over the set of terms of size {\tt N}.
\begin{code}
closedTerm(N,X):-closedTerm(X,0,N,0).

closedTerm(v(I),V)-->{V>0,V1 is V-1,between(0,V1,I)}.
closedTerm(l(A),V)-->l,{succ(V,NewV)},closedTerm(A,NewV).  
closedTerm(a(A,B),V)-->a,closedTerm(A,V),closedTerm(B,V).
\end{code}
The {\em size definition} is expressed by the work of the 
predicates {\tt l/1}, consuming {\em 1} size unit for
each {\em lambda binder} and {\tt a/2} consuming {\em 2} size units for each {\tt a/2} {\em application}
constructor and {\em 0} units for variables {\tt v/1}. 
The initial term which is just a unique variable has size {\em 0}.

Given that the number of leaves in a Motzkin tree is the number of binary nodes + {\em 1}, if follows that:
\begin{prop}
The set of terms of size $n$ for the size definition \{{\em application=2}, {\em lambda=1}, {\em variable=0}\}  is equal to the set of terms of size $n+1$ for
the size definition \{{\em application=1}, {\em lambda=1}, {\em variable=1}\}. 
\end{prop}
Thus our size definition gives the sequence {\bf A135501} of counts,
first introduced in \cite{lescanne13}, shifted by one.
For instance, the term $l(a(v(0),v(0)))$ will have size $3=1+2$ with our definition, which corresponds to
size $4=1+1+1+1$ using the size definition of \cite{lescanne13}.

Our size definition is implemented as
\begin{code}
l(SX,X):-succ(X,SX).
a-->l,l.
\end{code}
with  Prolog's DCG notation controlling the consumption of size units for {\tt N} to {\tt 0}.

The predicate {\tt toMotSkel/2} computes the Motzkin skeleton of a term. 
\begin{code}
toMotSkel(v(_),v).
toMotSkel(l(X),l(Y)):-toMotSkel(X,Y).
toMotSkel(a(X,Y),a(A,B)):-toMotSkel(X,A),toMotSkel(Y,B).
\end{code}
The predicate {\tt motSkel/2} generates Motzkin trees {\tt X} of size {\tt N}, using the same
size definition as the lambda terms for which they serve as skeletons.
\begin{code}
motSkel(N,X):-motSkel(X,N,0).

motSkel(v)-->[].
motSkel(l(X))-->l,motSkel(X).
motSkel(a(X,Y))-->a,motSkel(X),motSkel(Y).
\end{code}

\section{Closable and Unclosable Skeletons} \label{cs}

We call a Motzkin tree {\em closable}
if it is the skeleton of at least one closed lambda term.

The predicate {\tt isClosable/1} tests if it exists a closed lambda term
having {\tt X} as its skeleton. For each lambda binder it increments
a count {\tt V} (starting at {\tt 0}), and ensures that
it is strictly positive for all leaf nodes.
\begin{code}
isClosable(X):-isClosable(X,0).
  
isClosable(v,V):-V>0.
isClosable(l(A),V):-succ(V,NewV),isClosable(A,NewV).  
isClosable(a(A,B),V):-isClosable(A,V),isClosable(B,V).
\end{code}
We define generators for closable and unclosable skeletons
by filtering the stream of answers of the Motzkin tree generator
{\tt motSkel/2} with the predicate {\tt isClosable/1} and its negation.
\begin{code}
closableSkel(N,X):-motSkel(N,X),isClosable(X).

unClosableSkel(N,X):-motSkel(N,X),not(isClosable(X)).
\end{code}

It immediately follows that:
\begin{prop}\label{closkel}
A Motzkin tree is a skeleton of a closed lambda term if and only if
it exists at least one lambda binder on each path from the leaf to the root.
\end{prop}

In Fig. \ref{closing} we show 3 closable and 3 unclosable Motzkin skeletons.
\begin{figure}
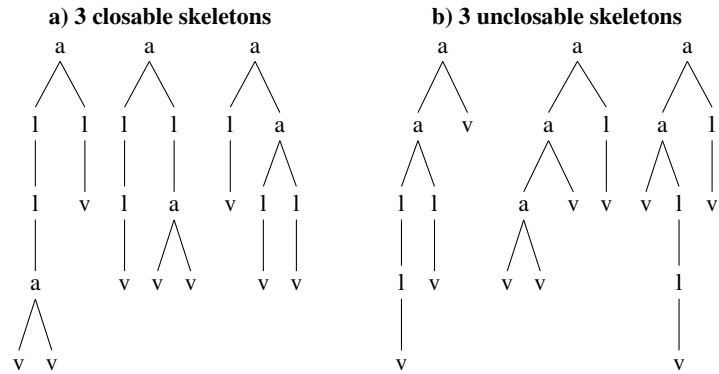

\begin{center}
{\bf a) 3 closable skeletons} ~~~~~~~~~~~~~~~~~~~~~~~~~ {\bf b) 3 unclosable skeletons}\\

\Tree [.a [.l [.l [.a [.v ] [.v ]  ]  ]  ] [.l [.v ]  ]  ]~
\Tree [.a [.l [.l [.v ]  ]  ] [.l [.a [.v ] [.v ]  ]  ]  ]~
\Tree [.a [.l [.v ]  ] [.a [.l [.v ]  ] [.l [.v ]  ]  ]  ]~~~~~~~~~~~~
\Tree [.a [.a [.l [.l [.v ]  ]  ] [.l [.v ]  ]  ] [.v ]  ]~
\Tree [.a [.a [.a [.v ] [.v ]  ] [.v ]  ] [.l [.v ]  ]  ]~
\Tree [.a [.a [.v ] [.l [.l [.v ]  ]  ]  ] [.l [.v ]  ]  ]

\end{center}
\caption{Closable vs. unclosable skeletons of size {\tt 7}}\label{closing}
\end{figure}

Next, we  derive the predicate {\tt quickClosableSkel/2} that
generates closable skeletons about 3 times faster
by testing directly that lambda binders are available at 
each leaf node, resulting in earlier
pruning of those that do not satisfy this constraint.
\begin{code}
quickClosableSkel(N,X):-quickClosableSkel(X,0,N,0).

quickClosableSkel(v,V)-->{V>0}.
quickClosableSkel(l(A),V)-->l,{succ(V,NewV)},quickClosableSkel(A,NewV).  
quickClosableSkel(a(A,B),V)-->a,
  quickClosableSkel(A,V),
  quickClosableSkel(B,V).
\end{code}

We observe that there are slightly more unclosable Motzkin trees than closable ones as size grows:
\vskip 0.5cm
\noindent
{\small {\em
closable:~~ 
0,1,1,2,5,11,26,65,163,417,1086,2858,7599,20391,55127,150028,410719, ...

\noindent
unclosable:
1,0,1,2,4,10,25,62,160,418,1102,2940,7912,21444,58507,160544,442748, ...
}}
\vskip 0.5cm

One step further, we can derive, based on Proposition \ref{closkel}, a
grammar generating closable skeletons, by observing that they require at least
one lambda ({\tt l/1} constructor) on each path, with Motzkin trees below the {\tt l/1}
constructor generated by the predicate {\tt motSkel/3} introduced in section \ref{ct}.
This runs about 3 times as fast as {\tt closableSkel/2}.

\begin{code}
closable(N,X):-closable(X,N,0).

closable(l(Z))-->l,motSkel(Z).
closable(a(X,Y))-->a,closable(X),closable(Y).
\end{code}

By entering this grammar as input to Maciej Bendkowski's Boltzmann sampler generator \cite{bend2017} 
we have obtained a Haskell program generating uniformly random closable skeletons
of one hundred thousand nodes in a few seconds. The probability to pick 
 {\tt l/1} and enter a Motzkin subtree instead of an {\tt a/2} constructor was {\tt 0.8730398709632761}.
 The  threshold within the Motzkin subtree to pick a leaf was {\tt 0.3341408333975344}, then
  {\tt 0.667473848839429} for a unary constructor, over which a binary constructor was picked. 
  See the {\bf Appendix} for the equivalent Prolog code.

Let us denote by $M(z)=\sum m_nz^n$ the ordinary generating function for Motzkin trees ($m_n$ is the number of Motzkin trees of size $n$). It is well known 
\cite{flajolet09} that $M(z)$ follows the algebraic functional equation
$M=z+zM+zM^2$ which can be obtained directly from the symbolic method and we get $M(z)={\frac {1-z-\sqrt {-3\,{z}^{2}-2\,z+1}}{2z}}$. From this, we obtain the classical result that asymptotically the number $m_n$ of Motzkin trees of size $n$ is equivalent to $\dfrac{\sqrt{3}}{2\sqrt{\pi}} 3^n n^{-3/2}$.

Now, following the proposition 1 (and the predicate {\tt closable/2} providing the corresponding grammar definition), we can deduce that the ordinary generating function $C(z)$ for closable lambda terms satisfies  $C(z)=zC(z)^2+zM(z)$. Indeed, a closable term has either an application at the root followed by two sub-closable terms (which gives rise to $zC(z)^2$), either an abstraction at the root followed by a term (which gives rise to $zM(z)$). Consequently, $C(z)={\frac {1-\sqrt {2\,z\sqrt {-3\,{z}^{2}-2\,z+1}+2\,{z}^{2}-2\,z+1
}}{2z}}$. Now, we are in the framework of the Flajolet-Odlysko transfer theorems \cite{flajolet1990singularity} which gives directly the asymptotics of the number $c_n$ of closable skeletons: $c_n \sim \dfrac{\sqrt{15}}{10\sqrt{\pi}} 3^n n^{-3/2}$. By dividing $c_n$ with  $m_n$ we obtain:

\begin{prop}
When $n$ tends to the infinity, the proportion of closable lambda term skeletons tends to $\dfrac{1}{\sqrt{5}}\doteq 44.7\%$.
\end{prop}

It is possible to calculate very efficiently the coefficients $c_n$. For that purpose, from the equation $C(z)=zC(z)^2+zM(z)$, an easy calculation gives that $C(z)$ satisfies the algebraic equation $z^2C(z)^4-2zC(z)^3+(-z^2+z+1)C(z)^2+(z-1)C(z)+z^2$. Thus, dealing with classical tools (in order to pass from an algebraic equation into a holonomic one), we can deduce a linear differential equation from it:
\begin{eqnarray*}
&0=-208{z}^{6}-168{z}^{5}+12{z}^{4}+94{z}^{3}-42{z}^{2
}+6z +\\ 
&\left( -16{z}^{6}+24{z}^{5}+36{z}^{4}-92{z}^{3}+60{z
}^{2}-12z \right) C \left( z \right) +\\
&\left( 768{z}^{9}-480{z}^
{8}-1088{z}^{7}-64{z}^{6}+216{z}^{5}+44{z}^{4}+30{z}^{3}-54
{z}^{2}+18z-2 \right) {\frac {\rm d}{{\rm d}z}}C \left( z \right) +\\
&\left( 384\,{z}^{10}-32\,{z}^{9}-368\,{z}^{8}-56\,{z}^{7}-4\,{z}^{6}
+110\,{z}^{5}-21\,{z}^{4}-21\,{z}^{3}+9\,{z}^{2}-z \right) {\frac {
{\rm d}^{2}}{{\rm d}{z}^{2}}}C \left( z \right)
\end{eqnarray*}
with the initial condition $C \left( 0 \right) =0.$
Now, extracting a relation on the coefficients from this holonomic equation, we obtain the following P-recurrence for the coefficient $c_n$: 
\begin{eqnarray*}
 &\left( 384{n}^{2}+384n \right) c_n +\\
 &\left( -32{n}^{2}-512n-480 \right) c_{n+1} +\\
 &\left( -368{n}^{2}-2192 n-2928 \right) c_{n+2} +\\ 
&\left( -56{n}^{2}-344n-504 \right) c_{n+3}+\\ 
&\left( -4{n}^{2}+188n+852 \right) c_{n+4} +\\ 
&\left( 110{n}^{2}+1034n+2328 \right) c_{n+5} +\\
&\left( -21{n}^{2}-201n-390\right) c_{n+6} +\\
&\left( -21{n}^{2}-327n-1272\right) c_{n+7} +\\ 
& \left( 9{n}^{2}+153n+648 \right) c_{n+8} + \\
& \left( -{n}^{2}-19n-90 \right) c_{n+9}=0\\
\end{eqnarray*}
with the initial conditions $c_0=0,
c_1=0,c_2=1,c_3=1,c_4=2,c_5=5,c_6=11,c_7=26,c_8=65 $

Note that by a guess-and-prove approach, we can a little simplify the recurrence into:

\begin{eqnarray*}
 &\left( 1200n^5+18480n^4+90816n^3+161088n^2+87552n \right) c_n +\\
 &\left( 800n^5+13520n^4+79024n^3+202312n^2+231768n+95760 \right) c_{n+1} +\\
 &\left( -100n^5-1840n^4-12848n^3-38792n^2-44100n-9576 \right) c_{n+2} +\\ 
&\left( -100n^5-1990n^4-14648n^3-48254n^2-66276n-23940 \right) c_{n+3}+\\ 
&\left( -225n^5-4815n^4-38883n^3-147519n^2-260286n-167580 \right) c_{n+4} +\\ 
&\left( 150n^5+3435n^4+29817n^3+120441n^2+218739n+131670 \right) c_{n+5} +\\
&\left( -25n^5-610n^4-5642n^3-24128n^2-45405n-26334\right) c_{n+6}=0\\
\end{eqnarray*}
with the initial conditions $c_0=0,
c_1=0,c_2=1,c_3=1,c_4=2,c_5=5$ 

This recurrence is extremely efficient in order to calculate the coefficient $c_n$.

Alternatively, the expansion into Taylor series of $C(z)$ gives $z^2 + z^3 + 2 z^4 + 5 z^5 + 11 z^6 + 26 z^7 + 65 z^8 + 163 z^9 + 417 z^{10} + 1086 z^{11} + 2858 z^{12} + 7599 z^{13} + 20391 z^{14} + 55127 z^{15} ... $ with its coefficients
matching the number of terms of sizes given by the exponents, corresponding to the number of solutions of the  predicate {\tt closableSkel/2}.  

\section{Uniquely Closable Skeletons} \label{uc}

We call a skeleton {\em uniquely closable} if it exists exactly one closed lambda term
having it as its skeleton.

\begin{prop}
A skeleton is uniquely closable if and only if exactly one lambda binder is available above
each of its leaf nodes.
\end{prop}
\begin{proof}
Note that if more than one 
were available for any leaf {\tt v}, one could choose more then one de Bruijn index 
at the corresponding leaf {\tt v/1} of a lambda term, resulting in more
than one possible lambda terms having the given skeleton.
\end{proof}

The predicate {\tt uniquelyClosable1/2} derived from {\tt quickClosableSkel1/2} ensures that
for each leaf {\tt v/0}  exactly one lambda binder is available. 

\begin{code}
uniquelyClosable1(N,X):-uniquelyClosable1(X,0,N,0).

uniquelyClosable1(v,1)-->[].
uniquelyClosable1(l(A),V)-->l,{succ(V,NewV)},uniquelyClosable1(A,NewV).  
uniquelyClosable1(a(A,B),V)-->a,uniquelyClosable1(A,V),
  uniquelyClosable1(B,V).
\end{code}

As a skeleton is uniquely closable if on any path
from a leaf to the root there's exactly one {\tt l/1} constructor, we derive the
predicate {\tt uniquelyClosable2/2} that marks subtrees below a lambda
{\tt l1/1} constructor to ensure no other {l/1} constructor is used
in them.

\begin{code}
uniquelyClosable2(N,X):-uniquelyClosable2(X,hasNoLambda,N,0).

uniquelyClosable2(v,hasOneLambda)-->[].
uniquelyClosable2(l(A),hasNoLambda)-->l,
  uniquelyClosable2(A,hasOneLambda).
uniquelyClosable2(a(A,B),Has)-->a,uniquelyClosable2(A,Has),
  uniquelyClosable2(B,Has).
\end{code}

By specializing with respect to having or not having a lambda binder above,
we obtain {\tt uniquelyClosable/2} which mimics a context-free grammar
generating all uniquely closable skeletons of a given size.
\begin{code}
uniquelyClosable(N,X):-uniquelyClosable(X,N,0).

uniquelyClosable(l(A))-->l,closedAbove(A).
uniquelyClosable(a(A,B))-->a,uniquelyClosable(A),uniquelyClosable(B).

closedAbove(v)-->[].
closedAbove(a(A,B))-->a,closedAbove(A),closedAbove(B).
\end{code}

In fact, if one wants to only count the number of solutions,
the actual term (argument 1) can be omitted, resulting in
the even faster predicate {\tt uniquelyClosableCount/1}.

\begin{code} 
uniquelyClosableCount(N):-uniquelyClosableCount(N,0).

uniquelyClosableCount-->l,closedAboveCount.
uniquelyClosableCount-->a,uniquelyClosableCount,uniquelyClosableCount.

closedAboveCount-->[].
closedAboveCount-->a,closedAboveCount,closedAboveCount.
\end{code}
This sequence of program transformations results in code running an order of
magnitude faster, with all counts  up to size {\tt 30}, shown in Fig. \ref{c_uc},
obtained in less than a minute.
Fig. \ref{c_uc} shows the growths of the set of uniquely closable skeletons.
\BF
x,y
0,0
1,1
2,0
3,1
4,1
5,2
6,2
7,7
8,5
9,20
10,19
11,60
12,62
13,202
14,202
15,679
16,711
17,2304
18,2507
19,8046
20,8856
21,28434
22,31855
23,101288
24,115596
25,364710
26,421654
27,1323946
28,1549090
29,4836072
30,5724582
\end{filecontents*}
\LPIC{size}
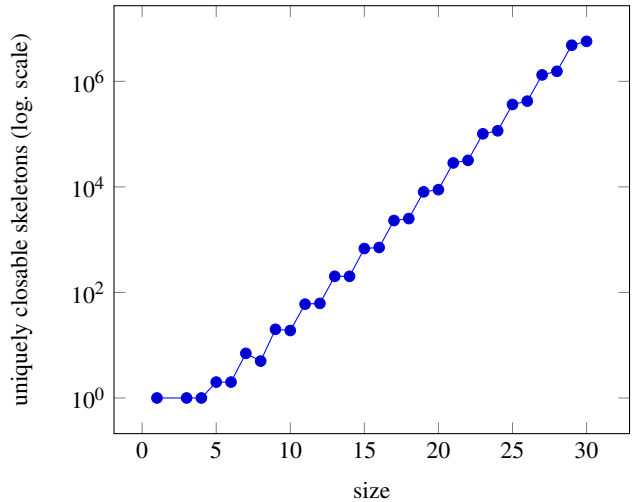
{uniquely closable skeletons}
{Uniquely closable skeletons by increasing sizes}{c_uc}

If expressed as a Haskell data type, the grammar describing
the set of closable skeletons becomes:
\begin{codex}
data UniquelyClosable =  L ClosedAbove 
  | A UniquelyClosable UniquelyClosable deriving(Eq,Show,Read)

data ClosedAbove = V | B ClosedAbove ClosedAbove deriving(Eq,Show,Read)
\end{codex}
With this notation, a skeleton, with the constructor {\tt B} used for
binary trees not containing an {\tt L} constructor, is
{\tt A (A (L V) (L V)) (L (B (B V V) V))}.

One can transliterate the Prolog DCG grammar into Haskell
by using list comprehensions to mimic backtracking as follows.
\begin{codex}
genA 0 = []
genA n | n>0 = 
   [L x | x <- genB (n-1)] ++ 
   [A x y | k <- [0..n-2], x <- genA k,  y <- genA (n-2-k)] 
   
genB 0 = [V]
genB n | n>0 = [B a b | k <- [0..n-2], a <- genB k, b <- genB (n-2-k)]   
\end{codex}

By entering the equivalent of this data type definition as input to Maciej Bendkowski's Boltzmann sampler generator \cite{bend2017} we have obtained a Haskell program generating uniformly random terms of one hundred thousand nodes in a few seconds. The probability threshold for a unary constructor was below {\tt 0.5001253328728457}
and then, once having entered a closed above subtree, it was {\tt 0.5001253328728457} to stop at a leaf rather than continuing with a binary node.  See the {\bf Appendix} for the equivalent Prolog code.

Let us denote by $B(z)$ the ordinary generating function for binary trees. The series $B(z)$ follows the algebraic functional equation
$B=z+zM^2$ and consequently $B(z)={\frac {1-\sqrt {-4\,{z}^{2}+1}}{2z}}$.

The ordinary generating function $U(z)$ for uniquely closable lambda terms satisfies  $U(z)=zU(z)^2+zB(z)$. Indeed, a uniquely closable term has either an application at the root followed by two sub uniquely closable terms (which gives rise to $zC(z)^2$), either an abstraction at the root followed by a term with no abstraction (which gives rise to $zB(z)$). Consequently, $U(z)=\frac {1-\sqrt {2\,z\sqrt {-4\,{z}^{2}+1}-2\,z+1}}{2z}$. We are again in the framework of the Flajolet-Odlysko transfer theorems \cite{flajolet1990singularity} which gives directly the asymptotics of the number $u_n$ of uniquely closable terms: $u_n \sim \frac {{2}^{1/4+n}}{4\Gamma  \left( 3/4 \right) {n}^{5/4}}$.

We can follow the same approach that for $C(z)$ to calculate quickly the coefficients $u_n$. In particular, $U(z)$ satisfies the algebraic equation $z^2U(z)^4-2zU(z)^3+(z+1)U(z)^2-U(z)+z^2=0$. From which we deduce a linear differential equation:
\begin{eqnarray*}
&0=-128z^5-40z^4+52z^3+18z^2-6z+(16z^5+56z^4-20z^3-20z^2+8z-2)U(z)+\\
&(512z^8-512z^7-320z^6+96z^5+144z^4+16z^3-24z^2-6z+2)({\frac {\rm d}{{\rm d}z}}U \left( z \right))+\\
&(256z^9-128z^8-128z^7-32z^6+64z^5+24z^4-16z^3-2z^2
+z)({\frac {
{\rm d}^{2}}{{\rm d}{z}^{2}}}U \left( z \right))\\
\end{eqnarray*}
with the initial condition $U \left( 0 \right) =0.$

Thus, we can efficiently compute the coefficient $u_n$ using the P-recurrence:
\begin{eqnarray*}
 &\left( 256n^2+256n \right) u_n ~+
 &\left( -128n^2-640n-512 \right) u_{n+1} ~+\\
 &\left( -128n^2-704n-880 \right) u_{n+2} ~+
&\left( -32n^2-64n+152 \right) u_{n+3} ~+\\ 
&\left( 64n^2+592n+1324 \right) u_{n+4} ~+
&\left( 24n^2+232n+540 \right) u_{n+5} ~+\\
&\left( -16n^2-200n-616\right) u_{n+6} ~+\\
&\left( -2n^2-32n-128\right) u_{n+7} ~+ 
&\left( n^2+17n+72 \right) u_{n+8}=0\\
\end{eqnarray*}
with the initial conditions $u_0 = 0, u_1 = 0, u_2 = 1, u_3 = 0, u_4 = 1, u_5 = 1, u_6 = 2, u_7 = 2 $

The Taylor series expansion of $U(z)$ gives $z^2 + z^4 + z^5 + 2 z^6 + 2 z^7 + 7 z^8 + 5 z^9 + 20 z^{10} + 19 z^{11} + 60 z^{12} + 62 z^{13} + 202 z^{14} + 202 z^{15} ...$ with coefficients of $z$ matching the number of solutions of the predicate {\tt uniquelyClosable/2} for sizes given by the exponents of $z$.

Let us notice that the polynomial factor in the asymptotics is not in $n^{-3/2}$ as it is universal for the tree-like structure. Here we have an interesting polynomial factor in $n^{-5/4}$ which appears when two square-root singularities coalesce. 

\section{Typable and Untypable Closable Skeletons} \label{utnt}
We call a Motzkin skeleton {\em typable} if it exists at least one 
simply-typed closed lambda term having it as its skeleton.
An {\em untypable} skeleton is a closable skeleton for which no such term exists.

We will follow the interleaving of term generation, 
checking for closedness and type inference steps shown in \cite{ppdp15tarau},
but split it into a two stage program, with the first stage generating code
to be executed, via Prolog's metacall by the second, while also
ensuring that the terms  generated by the second stage 
are closed.
 
The predicate {\tt genSkelEqs/4} generates type unification equations,
that, if satisfied by a closed lambda term, ensure that the term is
simply-typable.
\begin{code}
genSkelEqs(N,X,T,Eqs):-genSkelEqs(X,T,[],Eqs,true,N,0).

genSkelEqs(v,V,Vs,(el(V,Vs),Es),Es)-->{Vs=[_|_]}.
genSkelEqs(l(A),(S->T),Vs,Es1,Es2)-->l,genSkelEqs(A,T,[S|Vs],Es1,Es2).  
genSkelEqs(a(A,B),T,Vs,Es1,Es3)-->a,genSkelEqs(A,(S->T),Vs,Es1,Es2),
  genSkelEqs(B,S,Vs,Es2,Es3).

el(V,Vs):-member(V0,Vs),unify_with_occurs_check(V0,V).
\end{code}
Note that each lambda binder adds a new type variable to
the list (starting empty at the root) on the way down to a leaf node.
A term is then closed if the list of those variables {\tt Vs} is not
empty at each leaf node.

Thus, to generate the typable terms, one simply {\em executes the equations} {\tt Eqs},
as shown by the predicate {\tt typableClosedTerm/2}.
\begin{code}
typableClosedTerm(N,Term):-genSkelEqs(N,Term,_,Eqs),Eqs.
\end{code}

The predicate {\tt typableSkel/2} generates skeletons that are
typable by running the same equations {\tt Eqs} and ensuring they have {\em at least
one solution} using the Prolog built-in {\tt once/1}.
The predicate {\tt untypableSkel/2} succeeds, when the negation
of these equations succeeds, indicating that no simply-typed
lambda term exists having the given skeleton. 
Clearly, this is much faster than naively generating all the closed lambda 
terms and then finding their distinct skeletons.
\begin{code}
typableSkel(N,Skel):-genSkelEqs(N,Skel,_,Eqs),once(Eqs).
untypableSkel(N,Skel):-genSkelEqs(N,Skel,_,Eqs),not(Eqs).
\end{code}
In Fig. \ref{types} we show 3 typable and 3 untypable Motzkin skeletons.
\begin{figure}
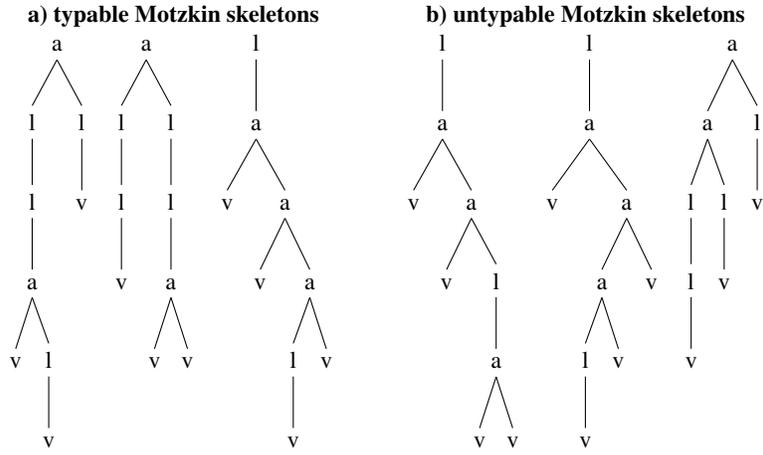

\begin{center}
{\bf a) typable Motzkin skeletons}~~~~~~~~~~~~~~~~~
{\bf b) untypable Motzkin skeletons}
\Tree [.a [.l [.l [.a [.v ] [.l [.v ]  ]  ]  ]  ] [.l [.v ]  ]  ]~
\Tree [.a [.l [.l [.v ]  ]  ] [.l [.l [.a [.v ] [.v ]  ]  ]  ]  ]~
\Tree [.l [.a [.v ] [.a [.v ] [.a [.l [.v ]  ] [.v ]  ]  ]  ]  ]~~~~~~~~~
\Tree [.l [.a [.v ] [.a [.v ] [.l [.a [.v ] [.v ]  ]  ]  ]  ]  ]~
\Tree [.l [.a [.v ] [.a [.a [.l [.v ]  ] [.v ]  ] [.v ]  ]  ]  ]~
\Tree [.a [.a [.l [.l [.v ]  ]  ] [.l [.v ]  ]  ] [.l [.v ]  ]  ]
\end{center}
\caption{Typable vs. untypable skeletons of size {\tt 8}}\label{types}
\end{figure}

An interesting question arises at this point about the relative density of closable and typable skeletons. Fig. \ref{c_vs_t}, shows how many typable skeletons are among the closable skeletons
for sizes up to {\tt 18}.
\BF
a,b,c
0,0,0
1,1,1
2,1,1
3,2,1
4,5,5
5,11,9
6,26,17
7,65,55
8,163,122
9,417,289
10,1086,828
11,2858,2037
12,7599,5239
13,20391,14578
14,55127,37942
15,150028,101307
16,410719,281041
17,1130245,755726
18,3124770,2062288
\end{filecontents*}
\LPICS{size}{Closable vs. typable skeletons}
{closable}{typable}
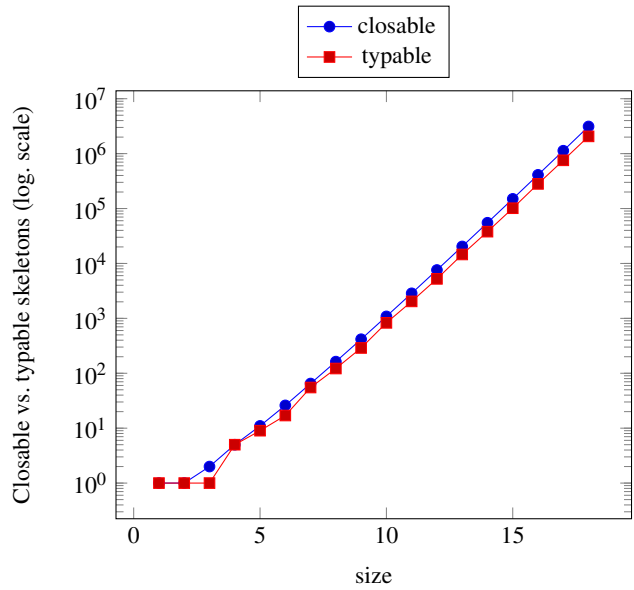
{Closable skeletons vs. typable skeletons by increasing sizes}{c_vs_t}
We leave as an {\em open problem} finding out the asymptotic behavior
of the relative density of the typable skeletons in the set of
closable ones.

\section{Uniquely Typable Skeletons and their Relation to Uniquely Closable Skeletons} \label{ucut}

A {\em uniquely typable skeleton} is one for which it exists exactly one
simply-typed closed lambda term having it as a skeleton.

\BF
x,y
0,0
1,1
2,0
3,0
4,2
5,0
6,1
7,7
8,1
9,13
10,34
11,20
12,100
13,226
14,234
15,853
16,1877
17,2650
18,8128
19,18116
20,30483
21,85713
\end{filecontents*}
\LPIC{size}
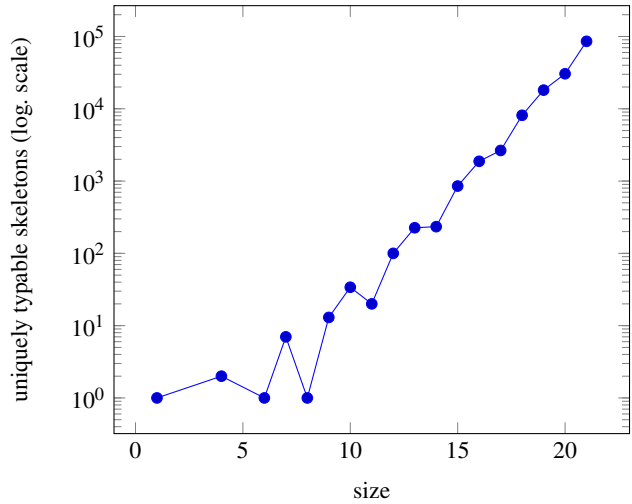
{uniquely typable skeletons}
{Uniquely typable skeletons by increasing sizes}{c_ut}

The predicate {\tt uniquelyTypableSkel/2} generates unification equations
for which, with the use of the built-in {\tt findnsols/4}, it ensures
efficiently that they have unique solutions. Fig. \ref{c_ut} shows the
counts of the skeletons it generates up to size {\tt 21}.
\begin{code}
uniquelyTypableSkel(N,Skel):-
  genSkelEqs(N,Skel,_,Eqs),has_unique_answer(Eqs).

has_unique_answer(G):-findnsols(2,G,G,Sols),!,Sols=[G].
\end{code}

The natural question arises at this point: are there (uniquely) typable skeletons
among the set of uniquely closable ones? The predicate {\tt uniquelyTypableSkel/2}
generates them by filtering the answer stream of  {\tt uniquelyClosable/2} with the
predicate {\tt isUniquelyTypableSkel/1}.
\begin{code}
uniquelyClosableTypable(N,X):-
  uniquelyClosable(N,X),isUniquelyTypableSkel(X).

isUniquelyTypableSkel(X):-skelType(X,_).
\end{code}
The predicate {isUniquelyTypableSkel/2} works by trying to infer the simple type
of a uniquely typable lambda term corresponding to the skeleton. Note that this
is a specialization of a general type inferencer to the case when exactly one
type variable is available for each leaf node.
\begin{code}
skelType(X,T):-skelType(X,T,[]).

skelType(v,V,[V0]):-unify_with_occurs_check(V,V0).
skelType(l(A),(S->T),Vs):-skelType(A,T,[S|Vs]). 
skelType(a(A,B),T,Vs):-skelType(A,(S->T),Vs),skelType(B,S,Vs).
\end{code}

\begin{prop}
Uniquely closable typable skeletons of size $3n+1$ are in bijection 
with Catalan objects (binary trees) of size $n$.
\end{prop}


\begin{proof}
We will exhibit a simple bijection to binary trees.
We want to show that terminal subtrees must be of the form l(v(0)).
As there's a unique lambda above each leaf, closing it, the leaf
should be (in de Bruijn notation), v(0) pointing to the first and only lambda
above it.
Assume a terminal node of the form a(v(0),v(0)). Then the two leaves must share
a lambda binder resulting in a circular term when unifying their types (i.e.,
 as in the case of the well-known term $\omega$=l(a(v(0),v(0))))
 and thus it could not be typable. 
\end{proof}

The following two trees illustrate the shape of such a skeletons and their  bijection to binary trees.
Note that the skeleton is mapped into a binary tree simply by replacing its terminal subtrees of the form
\verb~l(v)~ with a leaf node \verb~v~.
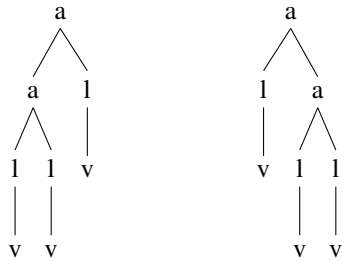
\begin{center}
\Tree [.a [.a [.l [.v ]  ] [.l [.v ]  ]  ] [.l [.v ]  ]  ]~~~~~~~~~~~~~~~~~~~~~
\Tree [.a [.l [.v ]  ] [.a [.l [.v ]  ] [.l [.v ]  ]  ]  ]
\end{center}
{\em In this case, terms of size $3n+1=7=2+2+1+1+1+0+0+0$ are mapped to binary trees of size $n=2=1+1+0+0+0$ (with {\tt a/2 nodes} there counted as 1 and v/0 nodes as 0) after replacing {\tt l(v)} nodes with {\tt v} nodes}.

As a consequence, {\em each uniquely closable term that is
typable is uniquely typable}, as identity functions
of the form {\tt l(v(0))} would correspond to the
end of each path from the root to a leaf in a lambda term having
this skeleton. 
This tells us that there no ``interesting'' uniquely
closable terms that are typable. However, as there are 
normalizable terms that are not simply typed, an
interesting {\em open problem} is to find out if closable terms,
other than those ending with {\tt l(v(0))}
are (weakly) normalizable.

\section{Related work}\label{rels}

The classic reference for lambda calculus is \cite{bar84}.
Various instances of typed lambda calculi are
overviewed in \cite{bar93}.

The first paper where de Bruijn indices are used in counting lambda terms is \cite{lescanne13}, which also
uses a size definition equivalent to ours (but shifted by 1).
The idea of using Boltzmann samplers for lambda terms was first introduced in \cite{binlamb}.

The combinatorics and asymptotic behavior of various
classes of lambda terms are extensively studied in 
\cite{grygielGen,bodini13}.
However, the concepts of closable and typable skeletons
of lambda terms and their uniquely closable and typable
variants are new and have not been studied previously.
The second author has used extensively Prolog as a meta-language 
for the study of combinatorial and computational properties of lambda 
terms in papers like \cite{padl17,ppdp15tarau} covering different families
of terms and properties, but not in combination with the precise
analytic methods as developed in this paper. 

It has been a long tradition in logic programming to
use step-by-step program transformations to derive
semantically simpler as well as more efficient code, going back as far as
\cite{pettorossi1994transformation}, that we have informally followed.
In \cite{fior15} a general constraint logic programming
framework is defined for size-constrained
generation of data structures as well as
a program-transformation mechanism. By contrast, 
we have not needed to use constraint solvers
in our code as our derivation steps allowed us
to place constraints  explicitly at the exact program points where
they were needed, for both the case of closable and typable
skeletons. Keeping our programs close to Horn Clause
Prolog has helped deriving CF-grammars for closable
and uniquely closable skeletons and  has enabled the
use of tools from analytic combinatorics to fully understand
their asymptotic behavior.

\section{Conclusion} \label{concl}

We have used simple program transformations to derive more efficient or
conceptually simpler logic programs in the process of attempting to state,
empirically study and  solve  some interesting problems related
to the combinatorics of lambda terms.
Several open problems  have also been generated in the process
with interesting implications on better understanding structural properties
of the notoriously hard set of simply-typed lambda terms.

The lambda term skeletons introduced in the paper involve abstraction mechanisms
that ``forget'' properties of the difficult class of simply-typed
closed lambda terms to reveal  classes of terms that
are  easier to grasp with analytic tools. 
In the case of the combinatorially simpler set of
closed lambda  terms, we have found that interesting
subclasses of their skeletons turn out to
be easier to handle. 
The case of uniquely closable terms turned out to be covered by a context free
grammar, after several program transformation steps, and thus amenable to
study analytically.

The focus on uniquely closable and uniquely typable Motzkin-tree skeletons of
lambda terms, as well as their relations, has shown that closability
and typability are properties that {\em predetermine} which lambda
terms in de Bruijn notation can have such Motzkin trees as skeletons.
Our analytic and experimental study has shown exponential growth for each of
these families and suggests possible uses as positive or negative lemmas
for all-term and random lambda term generation in dynamic programming algorithms.

Last but not least, we have shown that a language as simple
as side-effect-free Prolog, with limited use of impure features
and meta-programming, can  handle elegantly complex combinatorial 
generation problems, when the synergy between sound unification, 
backtracking and DCGs is put at work.

\section*{Acknowledgement} 
This research has been supported by NSF grant \verb~1423324~. 
We thank to Maciej Bendkowski for salient comments on an earlier draft of this paper
and the participants of the 10th Workshop on
Computational Logic and Applications 
( \url{https:/cla.tcs.uj.edu.pl} )
for enlightening discussions on the
combinatorics of lambda terms and their applications.

\bibliographystyle{../INCLUDES/splncs}
\bibliography{../INCLUDES/theory,../go/tarau,../INCLUDES/proglang}

\newpage
\section*{Appendix}

\subsection*{A Boltzmann sampler for closable lambda term skeletons}

\begin{code}
% generate uniformly random closable skeleton X
% of size between MinN and MaxN
genRanCL(MinN,MaxN):- 
  time(genRanCL(100000,MinN,MaxN,R,I)),
  writeln(R), % comment out for large sizes
  writeln(I),
  fail
; true.

% restart until a convenient term is generated
genRanCL(Tries,MinN,MaxN,X,I):-
  between(1,Tries,I),
  tryRanCL(MinN,MaxN,X),
  !.

% try to generate a term of size from MinN to MaxN
tryRanCL(MinN,MaxN,X):-
  random(R),
  ranCL(R,MaxN,X,MaxN,Dif),
  MaxN-Dif>=MinN.

% ensure random value R is below threshold
below(R,P,MaxN,N,N):-R=<P,N<MaxN.

% follow the grammar steps with given Boltzmann probabilities
% and pick lambda node with Motzkin trees below it 
% or subtrees with lambdas above them  
ranCL(R,MaxN,l(Z))-->below(R,0.8730398709632761,MaxN),!,
  l,
  {random(R1)},
  ranMot(R1,MaxN,Z).
ranCL(_,MaxN,a(X,Y))-->a,
  {random(R1),random(R2)},
  ranCL(R1,MaxN,X),
  ranCL(R2,MaxN,Y).

% generate Random Motzkin tree
ranMot(R,MaxN,v)-->below(R,0.3341408333975344,MaxN),!.
ranMot(R,MaxN,l(Z))-->below(R,0.667473848839429,MaxN),!,
  l,
  {random(R1)},
  ranMot(R1,MaxN,Z).
ranMot(_,MaxN,a(X,Y))-->a,
  {random(R1),random(R2)},
  ranMot(R1,MaxN,X),
  ranMot(R2,MaxN,Y).
\end{code}

\subsection*{A Boltzmann sampler for uniquely closable lambda term skeletons}

\begin{code}
% generate uniformly random uniquely closable skeleton X
% of size between MinN and MaxN
genRanUC(MinN,MaxN):- 
  time(genRanUC(100000,MinN,MaxN,R,I)),
  writeln(R), % to be commented out if too big
  writeln(tries=I),
  fail
; 
  true.

% restart until a convenient term is generated
genRanUC(Tries,MinN,MaxN,X,I):-
  between(1,Tries,I),
  tryRanUC(MinN,MaxN,X),
  !.

% try to generate a uniquely closable term 
% of size from MinN to MaxN
tryRanUC(MinN,MaxN,X):-
  random(R),
  ranUC(R,MaxN,X,MaxN,Dif),
  MaxN-Dif>=MinN.

% follow the grammar steps with given Boltzmann probabilities
% and pick unique lambda node or subtrees with lambdas in them
ranUC(R,MaxN,l(A))-->below(R,0.5001253328728457,MaxN),!,
  l,
  {random(R1)},
  ranCA(R1,MaxN,A).
ranUC(_,MaxN,a(A,B))-->a,
  {random(R1),random(R2)},
  ranUC(R1,MaxN,A),
  ranUC(R2,MaxN,B).

% generate binary subtree with no lambda in it
ranCA(R,MaxN,v)-->below(R,0.5001253328728457,MaxN),!.
ranCA(_,MaxN,a(A,B))-->
  a,
  {random(R1),random(R2)},
  ranCA(R1,MaxN,A),
  ranCA(R2,MaxN,B).
  
sampler_test1:-
  genRanCL(100000,200000).

sampler_test2:-
  genRanUC(100000,200000).
\end{code}

\subsection*{Some counts for the combinatorial objects of given sizes discussed in the paper, as
generated by code at \url{http://www.cse.unt.edu/~tarau/research/2017/uct.pro}
}

\BI
\I Motzkin Numbers: 1,1,2,4,9,21,51,127,323,835,2188,5798,15511, {\bf A001006} in \cite{intseq}
\I Closed Terms:  0,1,2,4,13,42,139,506,1915,7558,31092 $\ldots$ {\bf A135501} in \cite{intseq}
\I Closable Skeletons: 0,1,1,2,5,11,26,65,163,417,1086,2858,7599,20391,55127,\\ 
   150028,410719,1130245,3124770,8675210,24175809 $\ldots$
   \I Uniquely Closable Skeletons: 0,1,0,1,1,2,2,7,5,20,19,60,62,202,202,679,711,2304,\\
   2507,8046,8856,28434,31855,101288,115596,364710,421654,1323946,1549090 $\ldots$
\I Typable Closed Terms: 0,1,2,3,10,34,98,339,1263,4626,18099,73782,306295,\\
   1319660,5844714,26481404,123172740 $\ldots$
\I Typable Closable Skeletons:  0,1,1,1,5,9,17,55,122,289,828,2037,5239,14578,
   37942,\\101307,281041,755726,2062288 $\ldots$
\I Untypable Closable Skeletons: 0,0,0,1,0,2,9,10,41,128,258,821,2360,5813,\\17185,48721,129678,374519
\I Uniquely Typable  Skeletons: 0,1,0,0,2,0,1,7,1,13,34,20,100,226,234,\\
853,1877,2650,8128,18116,30483,85713 $\ldots$
\I Uniquely Closable Typable Skeletons: 0,1,0,0,1,0,0,2,0,0,5,0,0,14,0,0,42,0,0,132,\\
0,0,429,0,0,1430 $\ldots$
{\bf A000108}, {\em Catalan numbers} in \cite{intseq} for counts in position of the form $3n+1$.
\EI

\begin{codeh}

cm(N):-ncounts(N,motSkel(_,_)). 
bm(N):-ntimes(N,motSkel(_,_)).
sm(N):-nshow(N,motSkel(_,_)). 
pm(N):-npp(N,motSkel(_,_)). 
qm(N):-qpp(N,motSkel(_,_)).

cc(N):-ncounts(N,closedTerm(_,_)). 
bc(N):-ntimes(N,closedTerm(_,_)).
sc(N):-nshow(N,closedTerm(_,_)). 
pc(N):-npp(N,closedTerm(_,_)). 
qc(N):-qpp(N,closedTerm(_,_)).

ccs(N):-ncounts(N,closableSkel(_,_)). 
bcs(N):-ntimes(N,closableSkel(_,_)).
scs(N):-nshow(N,closableSkel(_,_)). 
pcs(N):-npp(N,closableSkel(_,_)). 
qcs(N):-qpp(N,closableSkel(_,_)).

ccs1(N):-ncounts(N,closable(_,_)). 
bcs1(N):-ntimes(N,closable(_,_)).
scs1(N):-nshow(N,closable(_,_)). 
pcs1(N):-npp(N,closable(_,_)). 
qcs1(N):-qpp(N,closable(_,_)).

cncs(N):-ncounts(N,unClosableSkel(_,_)). 
bncs(N):-ntimes(N,unClosableSkel(_,_)).
sncs(N):-nshow(N,unClosableSkel(_,_)). 
pncs(N):-npp(N,unClosableSkel(_,_)). 
qncs(N):-qpp(N,unClosableSkel(_,_)).

cqs(N):-ncounts(N,quickClosableSkel(_,_)). 
bqs(N):-ntimes(N,quickClosableSkel(_,_)).
sqs(N):-nshow(N,quickClosableSkel(_,_)). 
pqs(N):-npp(N,quickClosableSkel(_,_)). 
qqs(N):-qpp(N,quickClosableSkel(_,_)).

buc1(N):-ntimes(N,uniquelyClosable1(_,_)).
buc2(N):-ntimes(N,uniquelyClosable1(_,_)).

cuc(N):-ncounts(N,uniquelyClosable(_,_)). 
buc(N):-ntimes(N,uniquelyClosable(_,_)).
suc(N):-nshow(N,uniquelyClosable(_,_)). 
puc(N):-npp(N,uniquelyClosable(_,_)). 
quc(N):-qpp(N,uniquelyClosable(_,_)).

ct(N):-ncounts(N,typableClosedTerm(_,_)). 
bt(N):-ntimes(N,typableClosedTerm(_,_)).
st(N):-nshow(N,typableClosedTerm(_,_)). 
pt(N):-npp(N,typableClosedTerm(_,_)). 
qt(N):-qpp(N,typableClosedTerm(_,_)). 

cts(N):-ncounts(N,typableSkel(_,_)). 
bts(N):-ntimes(N,typableSkel(_,_)).
sts(N):-nshow(N,typableSkel(_,_)). 
pts(N):-npp(N,typableSkel(_,_)). 
qts(N):-qpp(N,typableSkel(_,_)). 

cuts(N):-ncounts(N,untypableSkel(_,_)).
buts(N):-ntimes(N,untypableSkel(_,_)).
suts(N):-nshow(N,untypableSkel(_,_)).  
puts(N):-npp(N,untypableSkel(_,_)).
quts(N):-qpp(N,untypableSkel(_,_)).

cus(N):-ncounts(N,uniquelyTypableSkel(_,_)). 
bus(N):-ntimes(N,uniquelyTypableSkel(_,_)).
sus(N):-nshow(N,uniquelyTypableSkel(_,_)). 
pus(N):-npp(N,uniquelyTypableSkel(_,_)). 
qus(N):-qpp(N,uniquelyTypableSkel(_,_)).

cucut(N):-ncounts(N,uniquelyClosableTypable(_,_)). 
bucut(N):-ntimes(N,uniquelyClosableTypable(_,_)).
sucut(N):-nshow(N,uniquelyClosableTypable(_,_)). 
pucut(N):-npp(N,uniquelyClosableTypable(_,_)). 
qucut(N):-qpp(N,uniquelyClosableTypable(_,_)).

go:-N=13,
  Tests=[cm,cc,ccs,cncs,cqs,cuc,ct,cts,cuts,cus,cucut],
  forall(member(F,Tests),call(F,N)).

fig1c:-fig1c(18).
fig1c(N):-time(ncvs2(N,quickClosableSkel(_,_),unClosableSkel(_,_))). 

fig2c:-fig2c(28).
fig2c(N):-time(ncvs(N,uniquelyClosable(_,_))).

fig1t:-fig1t(16).
fig1t(N):-time(ncvs2(N,typableSkel(_,_),untypableSkel(_,_))). 

fig2t:-fig2t(16).
fig2t(N):-time(ncvs(N,uniquelyTypableSkel(_,_))).

figct:-figct(16).
figct(N):-time(ncvs2(N,closableSkel(_,_),typableSkel(_,_))).

\end{codeh}

\end{document}